\documentclass[runningheads]{llncs}

\newif\ifarxiv
\arxivfalse
\arxivtrue

\newcommand{\cutout}[1]{\textcolor{green}{#1}}
\renewcommand{\cutout}[1]{}%

\ifarxiv
\fi

\usepackage{graphicx}
\usepackage[bookmarksnumbered,bookmarksopenlevel=2]{hyperref}

\usepackage{complexity}
\usepackage[dvipsnames]{xcolor}
\usepackage{xspace}
\usepackage{paralist} %
\usepackage[caption=false]{subfig}
\usepackage{amssymb}
\usepackage{amsmath}
\usepackage{cleveref}
\usepackage[normalem]{ulem}

\newtheorem{innercustomthm}{Theorem}
\newenvironment{customthm}[1]
{\renewcommand\theinnercustomthm{#1}\innercustomthm}
{\endinnercustomthm}

\newcommand{\mkmcal}[1]{\ensuremath{\mathcal{#1}}\xspace}

\renewcommand{\S}{\mkmcal{S}}
\renewcommand{\A}{\mkmcal{A}}

\newcommand{\etal}{{et al.}\xspace}
\newcommand{\sublab}[1]{{{(#1)}}}

\newtheorem{observation}{Observation}

\crefname{observation}{observation}{observation}
\Crefname{observation}{Observation}{Observation}
\crefname{lemma}{lemma}{lemma}
\Crefname{lemma}{Lemma}{Lemma}

\usepackage{todonotes}
\newcommand{\soeren}[2][noinline]{\personaltodo[#1]{RedOrange!30}{ST}{#2}}

\newcommand{\alexandra}[2][noinline]{\personaltodo[#1]{Orange!30}{AW}{#2}}

\newcommand{\guenter}[2][noinline]{\personaltodo[#1]{Dandelion!30}{GR}{#2}}

\newcommand{\newtextrev}[1]{#1}
\newcommand{\newtext}[1]{#1}
\newcommand{\removedtext}[1]{}%

\newcommand{\NkR}[1]{\textsf{N#1R}}
\newcommand{\gadvar}[1]{$\mathcal{N}(#1)$}
\newcommand{\gadedg}[1]{$\mathcal{L}(#1)$}
\newcommand{\bad}{popular}

\newcommand{\notbad}{unpopular}
\newcommand{\good}{\notbad}

\newcommand{\pnec}{\textsf{PNEC}\xspace}
\newcommand{\snesc}{\textsf{SNESC}\xspace}
\newcommand{\duplicate}{popular\xspace}

\newcommand{\arxivTHENecg}[2]{\ifarxiv#1\else#2\fi}

\ifarxiv
\newcommand{\refAppendix}[2]{Appendix~\ref{#1}}
\else
\newcommand{\refAppendix}[2]{the full version \cite[Appendix~#2]{denooijer2022removing}}
\fi

\newif\ifdoubleblind
\doubleblindfalse %
\newcommand{\removeDB}[2]{\ifdoubleblind#1\else#2\fi}

\graphicspath{{.},{Figures/}}

\begin{document}
\title{\texorpdfstring{\looseness=-1}{}
  Removing Popular Faces in Curve Arrangements\texorpdfstring{\thanks
  {Authors are sorted %
    by seniority. %
  }}{}
}

\newif\ifmergep
\mergeptrue
\newif\ifwebpage
\newif\ifnewlinewebpage
\newif\ifnewlineemail
\newif\ifonlyuniversity
\webpagefalse
\newlinewebpagefalse
\newlineemailfalse
\onlyuniversitytrue

\newcommand{\web}[1]{\ifwebpage\ifnewlinewebpage\\\else$\vert$ \fi #1\else\fi}
\newcommand{\mail}[1]{\ifnewlineemail\\\else$\vert$ \fi #1}
\newcommand{\uni}[2]{\ifonlyuniversity#1\else#2\fi}

\author{Phoebe de Nooijer\inst{1} \and
Soeren Terziadis\inst{2}\orcidID{0000-0001-5161-3841} \and
Alexandra Weinberger\inst{3}\orcidID{0000-0001-8553-6661} \and
Zuzana Mas{\'a}rov{\'a}\inst{4}\orcidID{0000-0002-6660-1322} \and
Tamara Mchedlidze\inst{\ifmergep 1\else 6\fi}\orcidID{0000-0002-1545-5580} \and
Maarten {L{\"o}ffler}\inst{\ifmergep 1\else 6\fi}\orcidID{0009-0001-9403-8856} \and
G{\"u}nter Rote\inst{\ifmergep 5\else 7\fi}\orcidID{0000-0002-0351-5945}
}

\authorrunning{%
  de Nooijer, %
  Terziadis, %
  Weinberger, %
  Mas\'arov\'a, %
  Mchedlidze, %
  Löffler, %
  Rote}

\institute{
	\uni{Utrecht University, Utrecht, the Netherlands}{Geometric Computing Group, Utrecht University, Utrecht, the Netherlands}
	\ifmergep \\
        \email{%
          m.loffler@uu.nl $\vert$ \email{t.mtsentlintze@uu.nl}} \web{\href{http://uu.nl/medewerkers/MLoffler}{uu.nl/medewerkers/MLoffler} $\vert$ \href{http://uu.nl/medewerkers/TMtsentlintze}{uu.nl/medewerkers/TMtsentlintze}}\else \mail\email{p.denooijer@students.uu.nl}\fi \and
	\uni{TU Wien, Vienna, Austria}{Algorithms and Complexity Group, TU Wien, Vienna, Austria}
	\mail{\email{soeren.nickel@ac.tuwien.ac.at}} \web{\href{http://ac.tuwien.ac.at/people/soeren-nickel/}{ac.tuwien.ac.at/people/soeren-nickel/}} \and
	\uni{Graz University of Technology, Graz, Austria}{Institute of Software Technology, TU Graz, Graz, Austria}
	\mail{\email{weinberger@ist.tugraz.at}} \web{\href{http://ist.tugraz.at/weinberger/}{ist.tugraz.at/weinberger/}} \and
	\uni{IST Austria, Maria Gugging, Austria}{Edelsbrunner Group, IST Austria, Maria Gugging, Austria}
	\mail{\email{zuzana.masarova@ist.ac.at}} \web{\href{http://research-explorer.app.ist.ac.at/person/45CFE238-F248-11E8-B48F-1D18A9856A87}{research-explorer.app.ist.ac.at/person/45CFE238-F248-11E8-B48F-1D18A9856A87}} \and
	\ifmergep\else \uni{Utrecht University, Utrecht, the Netherlands}{Visualization and Graphics Group, Utrecht University, Utrecht, the Netherlands}
	\mail{\email{t.mtsentlintze@uu.nl}} \web{\href{http://uu.nl/medewerkers/TMtsentlintze}{uu.nl/medewerkers/TMtsentlintze}} \and\fi
	\ifmergep\else \uni{Utrecht University, Utrecht, the Netherlands}{Visualization and Graphics Group, Utrecht University, Utrecht, the Netherlands}
	\mail{\email{m.loffler@uu.nl}} \web{\href{http://uu.nl/medewerkers/MLoffler}{uu.nl/medewerkers/MLoffler}} \and\fi
	\uni{Freie Universit\"at %
          Berlin, %
          Germany}{Freie Universit\"at Berlin, %
          Germany}
	\mail{\email{rote@inf.fu-berlin.de}}
	\web{\href{http://page.mi.fu-berlin.de/rote}{page.mi.fu-berlin.de/rote}}
}

\begingroup%
\ifarxiv
\renewcommand{\orcidID}[1]{\href{https://orcid.org/#1}{\includegraphics[scale=.03]{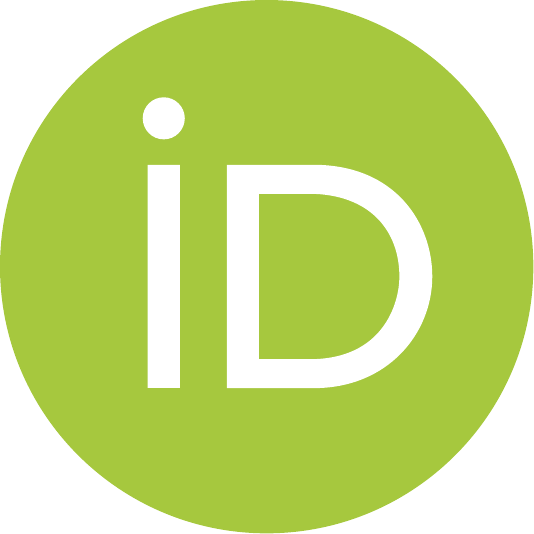}}}
\else
\def\orcidID#1{\href{https://orcid.org/#1}{\includegraphics[scale=.03]{orcid}}}
\fi
\maketitle
\endgroup %
\begin{abstract}
\looseness=-1
A face in a
curve arrangement is called \emph {popular} if it is bounded by the same curve multiple
times.   
Motivated by the automatic generation of curved nonogram puzzles,
we investigate possibilities to eliminate %
\newtext{the}
popular faces in an arrangement %
by inserting a single additional curve.  This turns out to be\cutout{\
  already}
\NP-hard; however,
\newtext{it becomes tractable when the number of popular faces is small:}
We present a probabilistic \FPT-approach
in the number of
\newtext{popular}
\removedtext{such}%
faces. %

\end{abstract}

\section{Introduction}

\begin {figure}[tbp]
\includegraphics[page=1]{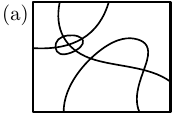}
\hfil
\includegraphics[page=2]{intro1}
\hfil
\includegraphics[page=3]{intro1}
\hfil
\includegraphics[page=4]{intro1}
\caption 
{ \looseness=-1
  \sublab{a} A curve arrangement in a
  rectangular
  frame.
	\sublab{b} The top right 
	face is incident to two disconnected segments of the
	red
	curve,
	making it \emph {popular}.
	\sublab{c} All popular faces are highlighted.
	\sublab{d} After inserting an additional curve, no more popular faces remain.
}
\label {fig:example}
\end {figure}

\looseness=-1
Let $\A$ be a set of curves which lie inside the area bounded by a closed
curve, %
called the \emph {frame}.
All curves in $\A$ are either \emph {closed},
or %
\emph {open}
with %
endpoints on the frame. %
We refer to $\A$ as a \emph {curve arrangement}, see
Figure~\ref {fig:example}a.
We consider only \emph{simple} arrangements, where no three curves
meet in a point and there are only finitely many total intersections, which are all crossings (no tangencies).

The arrangement $\A$ can be seen as an embedded multigraph whose
vertices are crossings of %
curves and whose edges are \emph {curve segments}.
$\A$ subdivides the region bounded by the frame %
into \emph
{faces}. A face is \emph {popular} when it is incident to multiple
curve segments belonging to the same curve in $\A$ (see Figures~\ref
{fig:example}b--c). %
We study the \newtext{\textsc{Nonogram 1-Resolution} (\NkR{1})}
problem: can one additional curve $\ell$ be inserted into $\A$ such
that no faces of $\A \cup \{\ell\}$ are popular (see Figure~\ref
{fig:example}d)?
\looseness-1

\subsubsection {Nonograms.} %

Our question is motivated by the problem of generating \emph {curved nonograms}.
Nonograms, also known as %
{\em Japanese puzzles}, {\em paint-by-numbers}, or {\em griddlers},
are a popular puzzle type where one is given an empty grid and a set of \emph {clues} on which grid cells need to be colored.
A clue consists of a sequence of numbers specifying the numbers of consecutive filled cells in a row or column. A solved nonogram typically results in a picture (see Figure~\ref {fig:nonograms_a}).
There is quite some work in the literature on the difficulty of solving nonograms~\cite {DBLP:journals/icga/BatenburgK12,DBLP:journals/dam/BerendPRR14,DBLP:journals/icga/ChenL19}.

\begin{figure}[tbp]
	\centering
	\subfloat[\label{fig:nonograms_a}]{%
		\includegraphics[page=1]{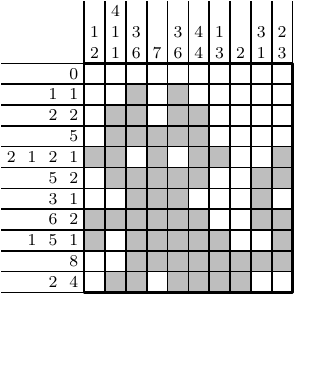}
	}\hfil
	\subfloat[\label{fig:nonograms_b}]{%
		\includegraphics[page=2]{nonograms}
	}
	\caption{Two nonogram puzzles in solved state.
		\sublab{a} A classic nonogram. 
		\sublab{b} A curved nonogram.}
	\label{fig:nonograms}
\end{figure}  

\looseness=-1
Van de Kerkhof~\etal~\cite {DBLP:journals/cgf/KerkhofJPLVK19} introduced {\em curved} nonograms, in
which the puzzle is no longer played on a grid but on an arrangement
of %
curves (see Figure~\ref {fig:nonograms_b}). 
In curved nonograms, clues specify numbers of filled faces of the arrangement in the sequence of faces incident to a common curve on one side.
Van de Kerkhof~\etal focus on heuristics to automatically generate such puzzles from a desired solution picture by extending curve segments to a complete curve arrangement.

\subsubsection {Nonogram complexity.}

Van de Kerkhof~\etal observed that curved nonograms come in different levels
of %
complexity --- not in terms of how hard it is to \emph {solve} a puzzle, but how hard it is to understand the rules (see Figure~\ref {fig:difficulty}).
They state that it would be of interest to generate puzzles of a
specific complexity level; their generators
\newtext{can currently do this only}
\removedtext{are currently not able to do so other than}%
by trial and error.
\begin {itemize}
\item \emph {Basic} nonograms are puzzles in which each clue
  corresponds to a sequence of
  \removedtext{unique}\newtext{distinct}
  faces. The analogy with clues in classic nonograms is straightforward.
\item \emph {Advanced} nonograms may have clues that correspond to a
sequence of faces in which some faces may appear multiple times
because the face is incident to the same curve (on the \emph{same} side)
multiple times. When such a face is filled, it is also counted
multiple times; in particular, it is no longer true that the sum
of the numbers in a clue is equal to the total number of filled
faces incident to the curve. This makes the rules harder to
understand.
\cutout{Advanced nonograms are not recommended for new users.}
\item \emph {Expert} nonograms may have clues in which a single face
  is incident to the same curve on \emph {both} sides.
  They are even more confusing than advanced nonograms. Expert
  nonograms are only suitable for %
experienced puzzle freaks. %
\end {itemize}	

It is easy to see that
arrangements with self-intersecting curves 
correspond exactly to
expert puzzles.
The difference between basic and advanced puzzles is more subtle; it is exactly the presence of \emph {popular faces} in the arrangement.

\begin {figure}[tbp]
\subfloat[]{%
	\includegraphics[page=1]{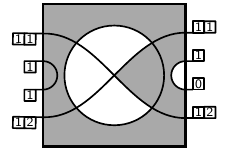}
}\hfil
\subfloat[]{%
	\includegraphics[page=2]{difficulty}
}\hfil
\subfloat[]{%
	\includegraphics[page=3]{difficulty}
}
\caption 
{
	Three types of curved nonograms of increasing
	complexity~\cite
	{DBLP:journals/cgf/KerkhofJPLVK19%
	}, shown with solutions.
	\sublab{a} \emph {Basic} puzzles have no popular faces.
	\sublab{b} \emph {Advanced} puzzles may have popular faces, but no self-intersections.
	\sublab{c} \emph {Expert} puzzles have self-intersecting curves.
	We can observe closed curves (without clues) in \sublab{a} and~\sublab{c}.
}
\label {fig:difficulty}
\end {figure}

One possibility to generate nonograms of a specific complexity would
be to take an existing generator and modify the output.
\newtext{Recently, Brunck et al.~\cite{buchin_et_al:DagRep.12.2.17}
  have investigated how \bad{} faces in a nonogram might be removed by
  reconfiguring and/or reconnecting parts of curves at small local
  areas, which they call switches (e.g. around curve crossings),
  and they have proved that this problem is \NP-hard.}
\newtext
{As an alternative, one may try to get rid of the \bad\ faces by
  adding extra
  curves that cut the \bad\ faces into smaller
  pieces.}
In this paper, we explore what we can do by inserting a single new curve into the \removedtext{output }arrangement.
Clearly, inserting\removedtext{ more}
curves will not remove self-intersections, so we focus on changing advanced puzzles into basic puzzles; i.e., removing all popular faces.

\begin {figure}	
\includegraphics {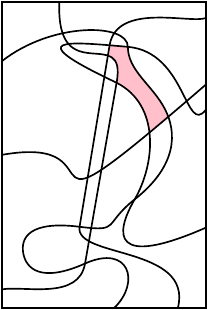}
\hfill
\includegraphics {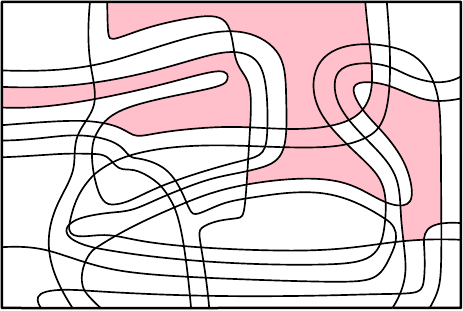}
\caption {Real puzzles (without clues) with all popular faces highlighted.}
\label {fig:puzzles}
\end {figure}

\subsection {Results.} %

After discussing
in Section~\ref{resolve-one-bad-face} how a singular face is resolved,
we show in Section~\ref {sec:hard} that deciding whether we can remove all popular faces from a given curve arrangement by inserting a single curve -- which we call the \NkR{1}~problem -- is \NP-complete.
However, often the number of popular faces is %
small, %
see Figure~\ref {fig:puzzles}.
Hence, we are also interested in the problem parametrized by the number of popular faces $k$. we show in Section~\ref {sec:cycle-edge} that the problem \textcolor{black}{can be solved by a randomized algorithm in \FPT{} time.}

\section{Resolving one \bad{} face by adding a single curve}
\label{resolve-one-bad-face}

As a preparation, we analyze how a single bad face $F%
$ can be resolved.
If $F%
$ is visited %
\newtext{three or more times}
by some
curve,
it
cannot be resolved with a single additional curve~$\ell$, and we can immediately abort.
Otherwise,
there are 
\emph{\duplicate} edges
among the 
edges 
of $F%
$, which belong
to a curve that visits $F%
$ twice.
As a visual aid, we
indicate each such pair of edges by connecting
them with a red
curve (a \emph{curtain}),
see Figure~\ref{fig:one-bad-face_a}
or%
~\ref{fig:app:vertexgadb}.
\begin{observation}
	\label{resolve-bad-face}
	To ensure that a \bad\ face $F$ becomes \good\ after insertion of a single
	curve~$\ell$ into the arrangement, it is necessary and sufficient that the curve~$\ell$ has the following properties.
	\begin{enumerate}
		\item It visits the face $F$ exactly once;
		\item It does not enter or exit through a \duplicate edge;
		\item It separates each pair of \duplicate edges. In other words, %
		$\ell$ cuts all curtains. \qed
	\end{enumerate}
      \end{observation}

\begin{figure}
	\centering
	\subfloat[\label{fig:one-bad-face_a}]{%
		\includegraphics[scale=.95,page=1]{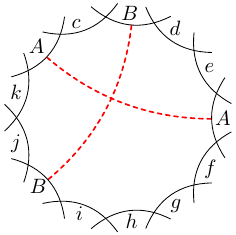}
	}\hfil
	\subfloat[\label{fig:one-bad-face_b}]{%
		\includegraphics[scale=.95,page=2]{bad-face}
	}\hfil
	\subfloat[\label{fig:one-bad-face_c}]{%
		\includegraphics[scale=.95,page=3]{bad-face}
	}
	\caption{Resolving a \bad\ face $F$.
          \newtext{(a) Curtains model %
            \duplicate edges.
     (b)~Possible ways how $\ell$ can pass through $F$.
     (c)~A more compact representation}
        }%
	\label{fig:one-bad-face}
\end{figure}

\looseness=-1
The ways how $\ell$ can traverse a \bad{} face $F$ can be modeled as a
graph: We place a vertex on every edge of~$F$ %
except the \duplicate edges.
We then connect two such vertices $u,v$ if for every curtain~$c$,
the endpoints of $c$ alternate with the vertices $u$ and $v$
around %
$F$, as shown in Figure~\ref{fig:one-bad-face_b}.
This representation can be condensed %
as shown in Figure~\ref{fig:one-bad-face_c} and explained in 
\refAppendix{sec:special-edge-sets}{D}.

\removedtext{
The blue segments in Figure~\ref{fig:one-bad-face_b} show
the %
ways how $\ell$
may pass through a %
\bad\ face.
}

In our \newtext{arguments}, %
we will %
\newtext{often}
use the dual graph $\A^d$ of a curve
arrangement~$\A$, where every face of $\A$ is represented by a vertex and edges represent faces which share a common boundary segment (not just a common crossing point).
In particular, %
a curve $\ell$ traversing $\A$ and crossing a sequence of faces $F_1,
\dots, F_k$ in that order can be expressed as a path $P = (F_1, \dots,
F_k)$ in~$\A^d$.

\section{\NkR{1} is \NP-complete}\label{sec:hard}
In order to prove \newtext{\NP-hardness, we reduce from
\emph{Planar Non-intersecting Eulerian Cycle}%
}.
This reduction assumes $\ell$ to be a closed loop, but %
it can easily be adapted to work for an open curve $\ell'$ starting
and ending at the frame.%

\subsection{Non-intersecting Eulerian cycles}

\newtext{
An Eulerian cycle %
in a graph is a closed walk that contains every edge
exactly once.
An Eulerian cycle in a graph embedded into the plane
(a plane graph)
is \emph{non-intersecting} if every
pair of consecutive edges
$(a,b),(b,c)$
is adjacent %
in the
radial order around~$b$. %
}
\newtext{Intuitively, an Eulerian cycle
  is non-intersecting if it can be drawn without
  repeated vertices %
after replacing each vertex
  by a small cycle linking the incident edges in circular
order (see Figures~\ref{fig:app:EulerianEx_a} and~\ref{fig:app:EulerianEx_b}). The Eulerian cycle has to visit all of the original edges, but it does
not have to cover the small vertex cycles (see
Figure~\ref{fig:app:EulerianEx_c}).}
\newtext{The following} %
problem was proved to be
\NP-complete by Bent and Manber~\cite[Theorem~1]{bent1987non}.

\begin{figure}[htp]
	\centering
	\subfloat[\label{fig:app:EulerianEx_a}]{%
		\includegraphics[page=4]{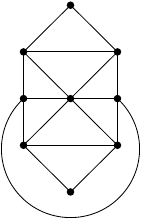}
	}
	\setcounter{subfigure}{10} %
	\subfloat{%
	\includegraphics[page=7,scale=1]{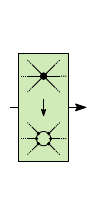}
	}
	\setcounter{subfigure}{1}
	\subfloat[\label{fig:app:EulerianEx_b}]{%
		\includegraphics[page=5]{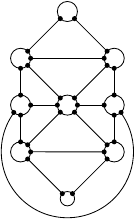}
	}\hfil
	\subfloat[\label{fig:app:EulerianEx_c}]{%
		\includegraphics[page=6]{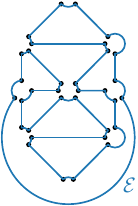}
	}
	\caption{
	Vertices of the graph $G$ \sublab{a} are replaced by cycles \sublab{b}. A non-intersecting cycle drawn on the modified graph visiting all original edges \sublab{c}.
	}
	\label{fig:app:EulerianEx}
\end{figure}

\begin{problem}[Planar Non-Intersecting Eulerian Cycle \pnec]\label{app:prob:NCEC}
  Given a
\newtext{planar graph embedded into the plane} %
  graph $G$, decide whether $G$ contains a non-intersecting Eulerian
  cycle. %
\end{problem}

\removedtext{Note that while this question is \NP-complete on general
  graphs, some trivial assumptions can be made on the input.}
\subsection{\NP-completeness reduction}\label{sec:app:hardness}

We will present a \removedtext{complete }polynomial-time reduction from \pnec to
\NkR{1}, i.e., we will create a curve arrangement $\A$ containing
\bad{} faces based on a planar input graph $G$ of \pnec, s.t.\ there exists a curve
$\ell$ \newtext{for which} %
$\A\cup\ell$ contains no \bad{}  faces if and only if
$G$ contains a non-intersecting Eulerian cycle. %
\newtext{We assume that $G$ is 2-edge-connected and
all vertices have even degree, because otherwise, $G$}
clearly cannot contain an Eulerian cycle.
We also replace every self-loop with a path of length two, without affecting the existence of a planar non-intersecting Eulerian cycle.

The reduction is gadget based.
We will represent every vertex $v \in V$ with a vertex gadget \gadvar{v} and every edge $e = (u,v) \in E$ with an edge gadget \gadedg{e} or \gadedg{u,v}.
Both gadgets are sets of curves starting and ending at the frame,
\newtext{and} %
$\newtext{\A} = \bigcup_{v\in V} \text{\gadvar{v}} \cup \bigcup_{e \in E}
\text{\gadedg{e}} %
$.

\subsubsection{Vertex gadgets.}\label{sec:app:vertex_gadget}
The vertex gadgets consist of curves in one of three basic shapes shown in Figure~\ref{fig:app:beakers}, which we call beakers.
  We place one beaker per incident edge of $v$, at the position
  of $v$, all rotated, s.t.\ their \emph{bases} (the lower
  ends in Figure~\ref{fig:app:beakers}) overlap
\removedtext{to form specific overlap patterns}%
\newtext{in a specific pattern}.
The \emph{opening} of each beaker (the upper ends in Figure~\ref{fig:app:beakers}) will point outwards.
We use three variants of the vertex gadget,
depending on the vertex degree.

\newtext{
The vertex gadget for a degree-two vertex is simply made up of two overlaying Type-I beakers (see Figure~\ref{fig:app:minigadget}).
Since $\ell$ must cross the two curtains %
$c_1$ and $c_2$, it must connect the two points $p_1, p_2$ by crossing the overlap of the two beakers (a face of degree two, marked in green in Figure~\ref{fig:app:minigadget}).
Since by Observation~\ref{resolve-bad-face}, $\ell$ can enter any beaker only once, the routing of $\ell$ as shown in Figure~\ref{fig:app:minigadgetresolve} is forced and corresponds exactly to the traversal of a planar non-intersecting Eulerian cycle through a vertex of degree two.
}

The vertex gadget \gadvar{v}
\newtext{for a degree-four vertex $v$}
consists of four Type-I beakers, 
one per incident edge, which\removedtext{ all
  pairwise intersect}
\cutout{have a common intersection
and }%
form the intersection pattern of
Figure~\ref{fig:app:smallgadget}.
\cutout
{The figure also shows %
  the four curtains of the %
  vertex gadget.}
\cutout
{We will refer to the green region in Figure~\ref{fig:app:smallgadget}
as the \emph{interior} of \gadvar{v}. \alexandra{Do we? I don't see interior used anymore.}}
\newtext{Since $\ell$ must cross the four curtains, it must enter
or exit the gadget at least four times through the thick blue edges in 
Figure~\ref{fig:app:smallgadgetdual} and the vertices
$p_1,p_2,p_3,p_4$ of the dual graph
$\A^d$.
Since $\ell$ cannot cross itself, there are only two possibilities how
$\ell$ can pass through \gadvar{v}:
It can connect
$p_1$ with $p_2$ and $p_3$ with $p_4$, as in
Figure~\ref{fig:app:smallgadgetresolve}, or $p_1$ with $p_4$ and $p_2$
with~$p_3$.
Both possibilities can be realized by routings of $\ell$, and 
they correspond precisely to the ways how a non-intersecting Eulerian
cycle
can pass through the edges incident to~$v$.}
Note that the exact routing of $\ell$ can vary inside \gadvar{v} (indicated by the dashed lines in Figure~\ref{fig:app:smallvertexgad}).

\begin{figure}[tbp]
	\centering
	\subfloat[\label{fig:app:beakers}]{%
		\includegraphics[page=9]{vertex-gadget-remake}
	}\hfil
	\subfloat[\label{fig:app:minigadget}]{%
		\includegraphics[page=6]{vertex-gadget2}
	}\hfil
	\subfloat[\label{fig:app:minigadgetresolve}]{%
		\includegraphics[page=7]{vertex-gadget2}
	}\\
	\subfloat[\label{fig:app:smallgadget}]{%
		\includegraphics{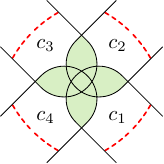}
	}\hfil
	\subfloat[\label{fig:app:smallgadgetdual}]{%
		\includegraphics[page=3]{vertex-gadget2}
	}\hfil
	\subfloat[\label{fig:app:smallgadgetresolve}]{%
		\includegraphics[page=4]{vertex-gadget2}
	}
	
	\caption{
		\sublab{a} Basic beaker curve shapes.%
		\newtext{
			\sublab{b} Degree two gadget (2 Type-I beakers) and \sublab{c} its forced resolution.
			\sublab{d} Degree four gadget (4 Type-I beakers).
			\sublab{e} Dual graph of the degree 4 gadget.}
		\sublab{f} A possible curve $\ell$ in \newtext{light blue};
		some alternative routings, which %
		connect the same endpoints, in \newtext{dashed light blue}.
	}
	\label{fig:app:smallvertexgad}
\end{figure}

\removedtext
{Any curtain in \gad{v} forces $\ell$ to cross the curtain and enter the interior of \gad{v}.
Since $\ell$ is a single curve (or in our case a closed loop), it cannot have any open ends.}
\removedtext{      
Therefore, the four open ends (blue nodes in Figure~\ref{fig:app:smallgadgetdual}) have to be connected to another open end.
Looking at the dual graph of \gad{v} (green edges and nodes in Figure~\ref{fig:app:smallgadgetdual}), we can see that any path from $p_1$, which enters the interior of \gad{v} and would connect to an open end, which is neither $p_2, p_3$ nor $p_4$, clearly has to contain \newtext{at least} one of these nodes.
Since every node in the cycle in $\A^d$ representing $\ell$ can only be adjacent to two edges in that cycle, any open endpoint of $\ell$ in \gad{v} can only be connected to another open endpoint in the same \gad{v}.
Further note that for any two paths $P_1 = (p_1, \dots, p_3)$ and $P_2 = (p_2, \dots, p_4)$, which enter the interior of \gad{v}, we know that $P_1 \cap P_2 \not= \emptyset$, i.e., two such paths necessarily cross.
Now there are only two valid options of connecting these endpoints, i.e., $p_1$ to $p_2$ and $p_3$ to $p_4$ (shown in Figure~\ref{fig:app:smallgadgetresolve}) or $p_1$ to $p_4$ and $p_2$ to $p_3$.
These correspond to a non-intersecting Eulerian cycle connecting the $i$-th edge to its right or left neighboring edge at $v$.}

The \removedtext{regular }vertex gadget \removedtext{used }for a vertex $v$ of
degree $
\removedtext{\delta(v)}\newtext{d} \geq 6$ is \removedtext{slightly }more complex.
We place $\removedtext{\delta(v)}\newtext{d}
-1$ Type-II beakers $c_1,\ldots,c_{
  \removedtext{\delta(v)}\newtext{d}
  -1}$ symmetrically around the location of $v$ (Figure~\ref{fig:app:vertexgada}).
Each beaker intersects four adjacent beakers (two on each side),
with the exception that $c_{d/2-1}$ and $c_{d/2+1}$ (\newtext{dark green} curves in Figure~\ref{fig:app:vertexgada}) do not intersect.
We place an additional Type-III beaker $c_d$ (the \newtext{light green} curve in Figure~\ref{fig:app:vertexgada})
that
surrounds all bases of the Type-II beakers
and protrudes between $c_{d-1}$ and $c_1$,
such that the intersection pattern of
Figure~\ref{fig:app:vertexgada} arises.

All \bad{} faces and curtains in \gadvar{v} are shown in Figure~\ref{fig:app:vertexgadb}.
The dual of the construction is shown in Figure~\ref{fig:app:vertexgadc}.
\newtext
{
  The curtains in the %
  green faces force $\ell$ to pass from these
  faces to the adjacent small faces with the blue boundaries.  This
  constrains $\ell$ to pass through a chain of faces as shown
  in
  Figure~\ref{fig:app:vertexgadd}.
  The passages from
  these faces to other neighboring faces can now be excluded,
  and the corresponding edges have been removed from the dual graph
  in
  Figure~\ref{fig:app:vertexgadd}.
}

The curtains in the openings of the beakers force the outer blue
endpoints of~$\ell$. \removedtext{ exactly as they did in the simple vertex gadget.}%
\newtext{The endpoint in beaker $c_i$ will be called~$p_i$.} Now we analyze \removedtext{ (similar to the simple vertex gadget),
}which of these endpoints can be connected with each other.
We see that in most cases,
  $p_i$ can only be
  connected to
  $p_{i-1}$ or $p_{i+1}$ without going through another endpoint.
The exception is   
$p_{d/2-1}$ and $p_{d/2+1}$, which can be connected via the inner loop
from $q_1$ to $q_2$.
However, this connection would cut off $p_{d/2}$ from the
remaining points.
We conclude that the visits of $\ell$ to \gadvar u must match
endpoints $p_i$ that are adjacent in the circular order.
There are two matchings,
which correspond to the two possibilities how a
non-intersecting Eulerian cycle can visit~$v$.
Both possibilities can be realized by routings of $\ell$; one is shown
in Figure~\ref{fig:app:vertexgade}, and the other is symmetric.
\cutout{END OF REPLACEMENT PROPOSAL.}

\cutout{
\newtext{Among these endpoints, an endpoint $p_i$ can only be
\soeren{Renamed to be more consistent with the previous case. Added labels to the figure}
  connected to
  $p_{i-1}$ or $p_{i+1}$,}
since any path to a different endpoint would contain one of these
endpoints.
\guenter{not true as written. p3 can go to p5 via the inner
  loop. suggested replacement text for the whole passage is inserted
  higher up.}
\soeren{I think the replacement text resolves this nicely. I would keep the transition sentence to the edge gadgets, but otherwise remove the old text. Marked as \textbackslash cutout (green, 16 lines)}
\alexandra{I agree with Soeren's comment including the suggested cutout. The new proposal is nice.}
The two inner endpoints $q_1, q_2$ cannot be connected to each other,
since $\ell$ has to be one single curve.
Also, any path from \newtext{$q_1$} to an endpoint, which is not \newtext{$p_{d/2}$ or $p_{d/2+1}$} (\newtext{$p_4$ and $p_5$} in Figure~\ref{fig:app:vertexgadd}), will contain \newtext{at least one other} endpoint and therefore \newtext{$q_1$} can only be connected to \newtext{$p_{d/2}$ or $p_{d/2+1}$}.
A similar argument holds for \newtext{$q_2$} the two endpoints \newtext{$p_{d/2-1}$ and $p_{d/2}$}.}

\cutout{
Finally assume that \newtext{$q_1$} and \newtext{$q_2$} are connected to \newtext{$p_{d/2+1}$} and \newtext{$p_{d/2-1}$} respectively.
In that case, any path connecting \newtext{$p_{d/2}$} with another endpoint has to contain \newtext{$p_{d/2-1}$ or $p_{d/2+1}$} and therefore this connection is also not possible.
This means either \newtext{$q_1$ or $q_2$} is connected to \newtext{$p_{d/2}$}.
In the first case \newtext{$q_2$} is connected to \newtext{$p_{d/2-1}$} and \newtext{$p_{d/2+1}$} to its \newtext{counterclockwise neighbor $p_{d/2+2}$} (\newtext{$p_6$} in Figure~\ref{fig:app:vertexgade}), and the latter case is symmetric.
Both cases again correspond to a non-intersecting Eulerian cycle connecting the $i$-th edge to its \newtext{clockwise or counter-clockwise} neighboring edge at $v$.
Note that in both cases, it is impossible for $\ell$ to exit \gad{v} other than over the outgoing blue segments inside the openings of the beakers.}

We now have placed vertex gadgets for all vertices. They require $\ell$ to connect to an endpoint in each opening of a beaker.
With these openings, we will now construct the edge gadgets.

\begin{figure}[tbp]
	\centering
	\subfloat[\label{fig:app:vertexgada}]{%
		\includegraphics[scale=.95,page=1]{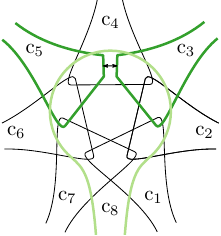}
	}\hfil
	\subfloat[\label{fig:app:vertexgadb}]{%
		\includegraphics[scale=.95,page=3]{vertex-gadget-remake}
	}\hfil
	\subfloat[\label{fig:app:vertexgadc}]{%
		\includegraphics[scale=.95,page=2]{vertex-gadget-remake}
	}\\
	\subfloat[\label{fig:app:vertexgadd}]{%
		\includegraphics[scale=.95,page=10]{vertex-gadget-remake}
	}\hfil
	\subfloat[\label{fig:app:vertexgade}]{%
		\includegraphics[scale=.95,page=6]{vertex-gadget-remake}
	}
	\caption{
		\sublab{a}~%
		Vertex gadget \gadvar{v} 
		for a degree-8 vertex $v$,
		\sublab{b}~its curtains and
		\sublab{c}~dual graph.
		\sublab{d}~Highlighted \newtext{light green} faces in \sublab{b} force the dark blue connections in the dual graph and restrict it. 
		\sublab{e}~One of two symmetric possibilities for the splitting curve~$\ell$.
		\newtext{The dashed lines} show different possible routings of $\ell$.
	}
	\label{fig:app:vertexgad}
\end{figure}

\begin{figure}[tbp]
	\centering
	\includegraphics[page=5]{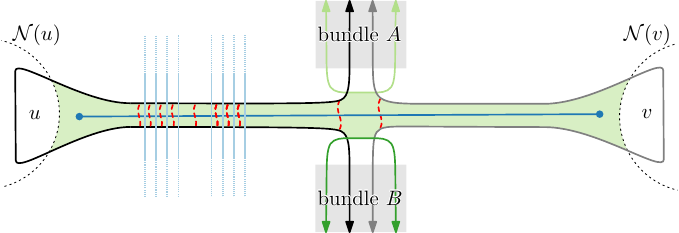} %
	\caption{
		Edge gadget \gadedg{u,v} connecting two beakers from \gadvar u
		and \gadvar v with \newtext{two} additional %
		curves.
		Open ends of all curves
		are collected into
		two bundles of parallel curves that lead into the incident faces.
		The inside of the two beakers are
		connected via a chain of \bad{} faces in \gadedg{u,v}
		(\newtext{shaded in light green}).
		Other bundles, \newtext{like the two groups of four light blue}
		curves
		\newtext{in the left half}, can freely cross either beaker.
	}
	\label{fig:app:edgegadget}
\end{figure}

\subsubsection{Edge gadgets.}\label{sec:app:edge_gadget}
Let $e = (u,v) \in E$ be an edge in $G$.
Then there are two vertex gadgets \gadvar{u} and \gadvar{v} already placed.
In particular, we placed one beaker in the gadgets per incident edge at $u$ or $v$, i.e. two per edge.

We now elongate the open ends of these beakers and route them
along the edge $e$ according to the embedding of $G$ given in the
input (recall that $G$ is a plane graph) until they almost meet at the center point of $e$.
 We bend the ends of each beaker outward,
\removedtext{to the left and one to the right,}%
routing them into the two faces %
incident to~$e$.
Additionally, we place two more curves on top (shown in green),
forming
the
intersection pattern of Figure~\ref{fig:app:edgegadget}.
This results in two  \emph{bundles} %
$A$ and $B$, each consisting
of four parallel curves. %
(The
light blue curves in the left half are not part of the gadget; %
they are two bundles that come from other gadgets.)

This connects a \bad{} face in \gadvar{u} to one in \gadvar{v}, forming one big \bad{} face in \gadedg{u,v}.
An arbitrary number of curves may cross the opening of a beaker.
The %
face is then simply divided into a chain of consecutive \bad{} faces.
\newtext
{In each of these faces,
except the
left- and right-most faces, which contain the blue endpoints,
$\ell$ has to leave
through two specific edges in order to cut the curtains.
This forces
$\ell$ to pass
straight through \gadedg{u,v}
from \gadvar u to \gadvar v
along the thin dark-blue horizontal axis.
}

It \newtext{remains} %
to describe how the open ends of the curves in the bundles are routed to the frame (since all curves other than $\ell$ have start and end at the frame).
\newtext{This is not difficult because these bundles can cross quite
  freely without creating \bad\ faces.
Each bundle consists of a unique set of curves, except for the two
bundles from one edge gadget.
  A bundle that originates from
the edge gadget \gadedg e
  can thus cross any bundle from a different edge gadget without creating
  popular faces.
It can cross a different
 edge gadget \gadedg {e'} by passing over
one of its
beakers, as shown with the light-blue curves.}

\newtext{Since $G$ does not contain self-loops, we route
each bundle along \newtext{a} path in the dual of $G$ to the outer
face of $G$, and then connect it to the frame, see
Figure~\ref{fig:frame_routing}.
}%
\newtext{A popular face might only be created when a bundle crosses the other bundle from the same edge gadget, which can be avoided by routing them in parallel.
}
The curves in a bundle \newtext{run in} %
parallel.
Two bundles originating from %
an edge gadget \gadedg e
have different curves
as their outside curves.
Hence, no \bad{} faces are created between
 two bundles, and we can make the following statement.

\begin{observation}\label{obs:all_bad_faces}

{All \bad{} faces in $\A$ are contained in
vertex gadgets and edge gadgets (the
faces \newtext{with dashed red curtains} in Figures~\ref{fig:app:minigadget}, \ref{fig:app:smallgadget},
\ref{fig:app:vertexgadb},
and~\ref{fig:app:edgegadget}).}

\begin{figure}[tb]
	\centering
	\includegraphics%
        {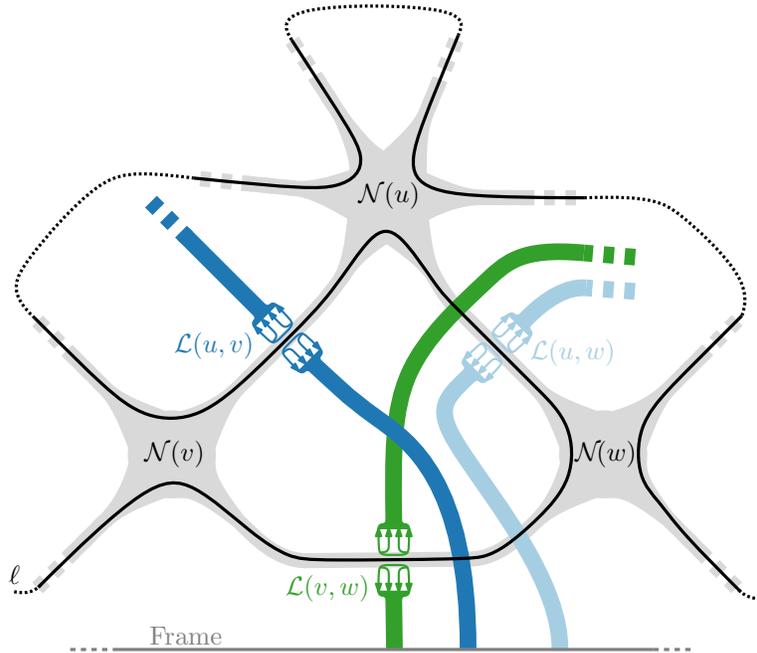}
	\caption{
		Schematic representation of three vertex and edge gadgets.
		The bundles are routed through beakers of other edges ending at the frame (partially shown at the bottom of the figure).
		A possible routing of $\ell$ is shown with a black curve.
	}
	\label{fig:frame_routing}
\end{figure}

\end{observation}

The next theorem follow from the construction and the resulting correspondence between resolving curves and non-crossing Eulerian cycles.
The proof can be found in
\refAppendix{sec:thm_hard_proof}{A}.
\begin{theorem}\label{thm:hardness}
  \NkR{1} is \NP-complete.
  \ifarxiv\else\qed\fi
\end{theorem}

\paragraph{Adaption to open curves. }
The reduction assumes that $\ell$ is a closed loop.
It can be adapted to work for open curves, i.e., we can create the arrangement $\A$, for which there exists an open curve $\ell'$ starting and ending at the frame, s.t.\ $\A\cup\ell'$ does not contain any \bad{} faces, if and only if $G$ contains a planar non-intersecting Eulerian cycle $\mathcal{E}(G)$.

\begin{figure}[htb]
	\centering
	\includegraphics[page=3]{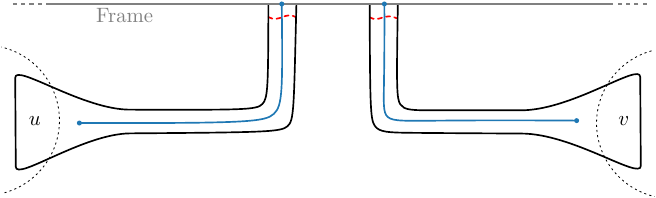} %
	\caption{
		By routing both ends of an open beaker in parallel to the frame, we force $\ell$ to start (or end) at the frame between the two connection points of the beaker.
	}
	\label{fig:app:edgegadget_open}
\end{figure}

The reduction creates $\A$ in the same fashion as above, except that we do not add the edge gadget for exactly one edge $e_o =(v_o, u_o)$ on the outer face of~$G$.
Instead, the open curves of the openings of the beakers in \gadvar{u_o} and \gadvar{v_o}, which would normally form the edge gadget \gadedg{u_o, v_o} are simply connected to the frame.
This forces $\ell'$ to start (and end) at the frame between the points at which the beakers connect to the frame, in order to properly split the \bad{} face in the opening of these beakers.
It is now easy to see that all other properties still hold and \NkR{1} remains \NP-complete even when $\ell'$ can be an open curve.

\section{\textcolor{black}{Randomized \FPT-algorithm for \NkR{1}}}
\label{sec:cycle-edge}

\looseness=-1
In this section, we
show that
\NkR{1} with $k$ \bad{} faces
can be solved by a randomized algorithm in $O\left(2^k\mathrm{poly}(n)\right)$
time, placing
\NkR{1} in the class randomized \FPT\ when parameterized by the number $k$ of \bad{} faces.
We model \NkR{1} as a problem of finding a simple
cycle
(i.e., a cycle without repeated vertices)
in a modified dual graph $G$, subject to a constraint that certain edges
must be visited.

\begin{problem}
  [Simple Cycle with Edge Set Constraints \snesc]
  \label{problem-edge-constraints}
	Given an undirected graph $G=(V,E)$ and $k$ %
	subsets
	$S_1,S_2,\ldots,S_k\subseteq E$ of edges, find a simple cycle, if it
	exists, that contains exactly one edge from each set $S_i$.
	\cutout{(The sets $S_i$ do not have to be disjoint.)}
\end{problem}

\newtext{ 
  We start with the dual graph of the given curve arrangement
  $\mathcal{A}$.  We replace the vertex corresponding to the $i$-th
  \bad{} face $f$ with a set $S_i$ of edges modeling the ways how an
  additional curve can cut all curtains of $f$, as described in
  Section~\ref{resolve-one-bad-face} and shown in
  Figure~\ref{fig:one-bad-face_b}. %
  \newtext{
    To be specific, we place a %
    vertex on each curve segment $s$ bounding $f$ and connect it to the vertex of the face that is adjacent to $f$ across $s$.
    Further we connect two such vertices on curve segments if a
    curve entering $f$ through one segment
    and exiting through the other would cut all curtains of $f$.
    The latter connecting edges, which run through $f$, form the set~$S_i$.
  }
  There is a one-to-one correspondence between the simple cycles containing
 exactly one edge of every set $S_i$ and the resolving
 curves for~$\mathcal{A}$.}

We will describe a randomized algorithm for\removedtext{ this problem,}
\newtext{Problem~\ref{problem-edge-constraints},}
extending
an algorithm of
Bj{\"o}rklund,
Husfeld, and Taslaman~\cite{Bjork2012}.
\begin{theorem}
	\looseness=-1	
  \begin{enumerate}[\rm(a)]
  \item 
	The \snesc problem
	on a graph with $n$
	vertices and $m\ge n$~edges can be solved in 
	$O\left(2^kmn^2\log\frac{2m}n\cdot
	|V(S_1)|\cdot W\right)$
	 time
	and $O(2^kn+m)$ space with a
	randomized Monte-Carlo algorithm,
	with probability at least $1-1/n^{\strut W}$, for any $W\ge 1$.
	Here, $V(S_1)$ denotes the set of vertices of the edges in $S_1$ (note that $S_1$ can be chosen to be the smallest set among $S_1,\ldots,S_k$).
	The model of computation is the Word-RAM with words of size
	$\Theta(k+\log n)$.
	
      \item 
	There is an alternative algorithm (described in \refAppendix{sec:poly-space}{B.3}) needing only polynomial space, namely
	$O(kn+m)$, at the expense of an additional factor $k$ in the
        runtime. It uses words of size 
	$\Theta(\log n)$.
\end{enumerate}
\smallskip

	Both algorithms find the cycle with the smallest number of edges
	if it exists \textup{(\!}with high probability\textup).

\end{theorem}

 \newtext
 {If~$\mathcal{A}$ has $n$ faces, the graph $G$ has $O(n)$ vertices
   and $m=O(n^2)$ edges.
   The quadratic blow-up of $m$ results from the construction as shown in
   Figure~\ref{fig:one-bad-face_b}.
\cutout{In practice, this is not an issue since the faces do not have
  many edges.}
The number of edges can be reduced to $O(n)$, as shown in
   Figure~\ref{fig:one-bad-face_c} and discussed in
   \refAppendix{sec:special-edge-sets}{D}.}
   The
  number $k$ of \bad\ faces is the same as the number $k$ of edge sets~$S_i$.

With the alternative algorithm with polynomial space, since %
$k\le n$, we get:

\begin{corollary}
  \label{main-algorithm}
  The N1R problem with $k$ \bad\ faces in a curve arrangement with
  $n$ faces can be solved in expected time
  $O(2^k \mathrm{poly}(n))$ and $O(kn)$ space.
  \qed
\end{corollary}

We first give a high-level overview of the algorithm.
We start by assigning random weights to the edges from a
sufficiently large finite field $\mathbb F_q$ of characteristic~2.
\newtext{Such a field %
  exists for every size~$q$ that is a power of~2.
  \newtextrev{In a field of characteristic~2, the law $x+x=0$ holds, and
    therefore terms cancel when they occur
an even number of times.}}
\textcolor{black}{The weight of a \cutout{directed }\emph{\newtext{walk}} (with vertex and edge repetitions allowed) is obtained by multiplying the edge weights of all visited edges}.
Our goal is now to compute the sum of weights all closed walks, of
given length, that
satisfy the edge set constraints.
The characteristic-2 property will ensure that the unwanted walks,
those which are not simple, cancel, while a simple
 closed walk makes a nonzero contribution and leads to a nonzero
 sum with high property.
 The crucial idea is that,
while these sets of closed walks can be very complicated,
 we can compute 
the aggregated sum of their weights in polynomial time.
We have to anchor these walks at some starting vertex~$b$,
and we choose $b$ to be one of the vertices incident to an edge of $S_1$.

More precisely,
for each such vertex $b$, %
and for
increasing lengths $l=1,2%
,\ldots,n$, the algorithm computes the quantity
$\hat T_b(l)$, which is
the sum of the weights
of all closed
walks that
\begin{itemize}
	\item start and end at $b$,
	\item have their first edge in $S_1$,
	\item use exactly one edge from each set $S_i$ (and
	use it only once),
	\item and consist of $l$ edges.
\end{itemize}
\cutout{If no such walks exist, $\hat T_b(l)=0$.}

We consider the edge weights as variables %
and regard $\hat T_b(l)$
as a function of these variables.
The result is a
 polynomial %
where each term is a product of
$l$ variables (possibly with repetition), and hence the polynomial has
degree~$l$, unless all terms cancel and it is the zero polynomial.
We %
 apply the following lemma, which is a straightforward
adaptation of a lemma of
Bj{\"o}rklund et al.~\cite{Bjork2012}.%
\cutout{, and whose proof is given
in Section~\ref{shortest-L-nonzero-proof}.}%

\begin{lemma}
  \label{shortest-L-nonzero}
  \begin{enumerate}[a)]
  \item 
\removedtext{Let $L$ be the length of a shortest simple cycle}
\newtext
{Suppose there exists a simple cycle of length $l$} %
that satisfies the edge
set constraints and that goes through an edge %
of $S_1$
\removedtext{that has $b$ as one of its endpoints.}%
\newtext
{incident to~$b$.}
Then the polynomial $\hat T_b(l)$ %
is homogeneous of degree $l$ and is not identically zero.
\item \label{longer}
  \newtext
  {If there is no such cycle of length $\le l$, %
    the polynomial $\hat T_b(l)$ is identically zero.}
\end{enumerate}
\end{lemma}

\newtext{
  \cutout{PROPOSED SENTENCES TO SKETCH THE PROOF.}
The lemma is based on the fact that
each term in the polynomial $\hat T_b(l)$ represents some closed walk.
A term coming from a walk that visits a %
vertex twice
can be matched
with another walk, which traverses a loop in the opposite direction and
contributes the same term. Since the field has
characteristic~2,
these terms
cancel.
A term coming from a simple walk does not cancel.
The %
proof
of Lemma~\ref{shortest-L-nonzero} %
is given
in \refAppendix{shortest-L-nonzero-proof}{B.5}.
}

\newtext{In case (a) of Lemma~\ref{shortest-L-nonzero}},
it follows from the Schwartz-Zippel
Lemma~\cite[Corollary~1]{%
  DBLP:journals/jacm/Schwartz80}
\removeDB{}
{(see also
\cite[Corollary~Q1]{nullstellensatz})}
that, for randomly chosen
weights in~$\mathbb F_q$, $\hat T_b(l)$ is nonzero with probability at least
$1-\text{degree}/|\mathbb F_q| = 1-l/q\ge 1-n/q$.

Thus, if we choose $q>n^2$, we have a success probability of at least
$1-1/n$ for finding the shortest cycle when we %
evaluate the quantities
$\hat T_b(l)$ for increasing~$l$ until they become nonzero.
The success probability can be boosted by %
repeating the experiment with new random weights.

\newtext{In the unlikely case of a failure},
the algorithm may err by not finding a solution although a solution exists, or by finding
a solution\removedtext{ of length $l$} that is not shortest.
\newtextrev{The last possibility is not an issue for our original
problem, where we just ask about the existence of a cycle, of
arbitrary length.}
\removedtext
{(%
By  construction,
  any solution that
is found will always be a simple cycle that satisfies the edge set
constraints.)}

In the following, we will
discuss how we can compute the
quantities 
$\hat T_b(l)$, and the runtime and space requirement for this
calculation.
We
describe
the method for actually recovering the cycle after we have found
a nonzero value in \refAppendix{recover}{B.4}.

\subsection{Computing sums of path weights by dynamic programming}
\label{dyn-prog}

We cannot compute the desired sums $\hat T_b(l)$ directly, but have to do this
incrementally
via a larger variety of quantities 
$ T_b(R,l,v)$
that are defined
as follows:

For $R\subseteq\{1,2,\ldots,k\}$ with $1\in R$, $v\in V$, and $l\ge
1$, we define
$ T_b(R,l,v)$ as
the sum of the weights
of all walks that
\begin{compactitem}
	\item start at $b$,
	\item have their first edge in $S_1$,
	\item end at $v$,
	\item consist of $l$ edges,
	\item use exactly one edge from each set $S_i$ with $i\in R$ (and
	use it only once),
	\item contain no edge from the sets $S_i$ with $i\notin R$.
\end{compactitem}
The walks that we consider here differ from the walks in $\hat T_b(l)$
in two respects: They end at \removedtext{an arbitrary}%
\newtext{a specified} vertex $v$, and the set $R$ keeps
track of
the sets $S_i$ that were already visited.
The quantities $\hat T_b(l)$ that we are interested in arise as
a special case when we have visited the full range
$R=\{1,\ldots,k\}$ of sets~$S_i$ and arrive at $v=b$:
\begin{displaymath}
	\hat T_b(l) =  T_b(\{1,\ldots,k\},l,b).
\end{displaymath}

We compute the values $T_b(R,l,v)$
for increasing values \newtext{$l=1,\ldots,n$}.
The starting values for $l=1$ are straightforward from the definition.

To compute $T_b(R,l,v)$ \newtext{for $l\ge2$}, we collect all stored values of the form
$T_b(R',l-1,u)$ where $(u,v)$ is an edge of $G$ and $R'$ is derived
from $R$
by taking into account the sets $S_i$ to which $(u,v)$ belongs.
We multiply these values with \newtext{the edge weight $w_{uv}$}
and sum them up.
If $(u,v)$ is in some
$S_i$ but $i \notin R$, we don't use this edge.
Formally, %
let $I(u,v) := \{\,i\mid (u,v)\in S_i\,\}$ be the index set of
the sets $S_i$ to which $(u,v)$ belongs. Then
\begin{equation}
	\label{eq:recursion}
	T_b(R,l,v) =
	\sum_
	{\substack{(u,v)\in E\\
			I(u,v)\subseteq R}} w_{uv} \cdot T_b(R\setminus I(u,v),l-1,u)
\end{equation}

\subsection{Runtime and space}
\label{sec:runtime}

\begin {figure}[tbp]
\subfloat[]{%
	\includegraphics[page=1]{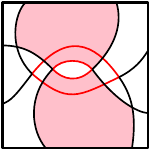}
}\hfil
\subfloat[]{%
	\includegraphics[page=2]{implementation}
}\hfil
\subfloat[]{%
	\includegraphics[page=3]{implementation}
}\hfil
\subfloat[]{%
	\includegraphics[page=4]{implementation}
}
\caption {Input \sublab{a,c} and resulting output \sublab{b,d} generated by
	the implementation; the green curve
	resolves the popular
	faces
	with the smallest number of crossings.
}
\label{fig:implementation}
\end {figure}

\newtext{
  The finite field additions and multiplications %
  in~\eqref{eq:recursion}
  take constant time, see \refAppendix{sec:field}{B.2}.
  Similarly,
  the set operation $R\setminus I(u,v)$
  on subsets of
  $\{1,\ldots,k\}$
  and the
  test $ I(u,v)\subseteq R$
  can be carried out in constant time,
  using bit vectors.
Thus,
  for a fixed
 starting vertex $b\in V(S_1)$ and fixed $R$,
going from $l-1$ to $l$ by the recursion~\eqref{eq:recursion}
takes
$O(m)$ time in total, because each edge $(u,v)$ appears
  in at most one of the sums on the right-hand side.
  The overall runtime is
$O(|V(S_1)|2^kmn)$.}

\newtext{
  As mentioned after Lemma~\ref{shortest-L-nonzero},
  the probability that
  the algorithm misses the shortest simple path is at most
$1/n$.
  To amplify the probability of correctness, we repeat the computation
  $W$ times,
  reducing the failure probability
to $1/n^W$.}

\looseness=-1
We consider each starting vertex $b$ separately, and do not need to store
\newtext{entries for}
lengths $l-1$
\newtext{or shorter}
when proceeding from $l$ to $l+1$;
thus the space requirement is $O(2^kn)$.

For recovering the solution,
the runtime must
be multiplied by $O(n\log\frac{2m}n)$,
see \refAppendix{recover}{B.4}.
Figure~\ref{fig:implementation} shows initial results of an implementation
of our algorithm on two small test instances;
see
also~\cite {thesis/phoebe}.%

\section{Conclusion}
In light of our \NP-hardness and randomized \FPT-algorithm, a natural next step is a deterministic parameterized algorithm.
There are $O(n)$ local possibilities of resolving a single \bad{} face, however, this does not immediately lead to an $O(n^k)$ algorithm (which would place \NkR{1} in \XP), since we might need to branch additionally over all possible connections between these solutions through the dual of $\A$, which can have an unbounded size.

\paragraph{Acknowledgements.}
This work was initiated at the 16th European Research Week on Geometric Graphs in Strobl in 
2019.
A.W. is supported by the Austrian Science Fund (FWF): W1230.
S.T. has been funded by the Vienna Science and Technology Fund (WWTF) [10.47379/ICT19035].
A preliminary version of this work has been presented at the
 38th European Workshop on Computational Geometry (EuroCG 2022) in
 Perugia~\cite{nnwmmlr-rpfca-22}.
\ifarxiv
\else
A full version of this paper,
which includes %
appendices but is
otherwise identical, is
available as a technical report~\cite{denooijer2022removing}.
\fi

\addcontentsline {toc}{section}{References}

\bibliographystyle{splncs04}
\bibliography{impop}

\arxivTHENecg{}{\end{document}}

\newpage
\appendix

\newcommand{\tobeincluded}{}

\section{Proof of Theorem~\ref{thm:hardness}}\label{sec:thm_hard_proof}

\begin{customthm}{1}
	\NkR{1} is \NP-complete.
\end{customthm}
\begin{proof}
	Assume we are given a non-intersecting Eulerian cycle $\mathcal{E}(G)$ of $G$ as a permutation of the edges.
	We now show how to construct the curve $\ell$.
	We choose a random edge $e = (u,v)$ in the permutation and start drawing $\ell$ at the endpoint inside the beaker of \gadedg{u,v} originating from \gadvar{v} along the thin blue axis crossing all \bad{} faces as described above.%
	This will end at an open endpoint inside \gadvar{u}.
	At this point we either connect to the left or right endpoint according to the next edge in $\mathcal{E}(G)$  (which is possible, since $\mathcal{E}(G)$ is non-intersecting). We make this connection either directly or via the forced inner circular part of $\ell$ in \gadvar{u} if one of the two involved endpoints is in beaker $c_{d(u)/2}$ (where $d(u)$ is the degree of $u$).
	These connections are made as described above.%
	Now we are again at an endpoint in a beaker.
	We can repeat this procedure until we cycle back to $e$, at which point we will have reconnected to our starting point.
	Since $\mathcal{E}(G)$ by definition visits every edge exactly once, before cycling back to the first edge, we know that every \bad{} face in the edge gadgets is split.
	Moreover, since $\mathcal{E}(G)$ visits every vertex $v$ exactly $\delta(v)/2$ times and $\mathcal{E}(G)$ is non-intersecting, we know that one of the two possible ways of connecting all endpoints in \gadvar{v} can be chosen to resolve all \bad{} faces in \gadvar{v} (and all other vertex gadgets).
	Since by Observation~\ref{obs:all_bad_faces} there are no other \bad{} faces in $\A$, $\A\cup\ell$ does not contain any \bad{} faces.
	
	Now assume we are given a curve $\ell$, s.t.\ $\A\cup\ell$ does not contain any \bad{} faces.
	The order, in which $\ell$ traverses all edge gadgets gives us a permutation of all edges.
	Since both variants of resolving all \bad{} faces in a vertex gadget connect any endpoint only to an endpoint in a neighboring beaker, consecutive edges in this permutation share an endpoint and both belong to the same face and this permutation is a non-intersecting Eulerian cycle.
	We have shown that there exists a curve $\ell$, s.t., $\A\cup\ell$ contains no \bad{} faces, if and only if $G$ contains a non-intersecting Eulerian cycle and therefore \NkR{1} is \NP-hard.
	
	\removedtext{        
		We represent $\A \cup \ell$ as a graph by introducing vertices at intersection points of curves and connecting vertices corresponding to consecutive intersection points along curves.
		We call such a graph a certificate of \NkR{1}.
		Such a certificate is clearly of polynomial size (polynomially balanced).
		Moreover, we can check in polynomial time, if such a graph contains faces, which are bounded by two or more edges, which originated from the same curve (polynomially verifiable).
		These two criteria are sufficient to show \NP-containment.
	}
	\newtext      
	{It is easy to see that \NkR{1} is in \NP. The input arrangement can be
		represented as a plane graph in which the edges are marked as
		belonging to the different curves or to the frame boundary.
		The certificate is an extension of this arrangement by
		a resolving
		curve~$\ell$. 
		Since $\ell$ cannot visit a face more than once, the certificate is
		of polynomial size.
		It can be easily verified in
		polynomial time whether it is valid,
		in particular, whether it %
		contains no \bad\ faces.
	}  
\qed
\end{proof}

\newif\ifstandalone
\ifdefined\tobeincluded
\else
\standalonetrue

\documentclass[11pt,a4paper]{article}

\usepackage[margin=2cm]{geometry}
\usepackage{amsmath}
\usepackage{amsthm}
\usepackage{amsfonts}
\usepackage{graphicx}
\usepackage{hyperref}
\usepackage{paralist} %
\usepackage[utf8]{inputenc}

\usepackage{complexity}

\newcommand{\newtext}{}
\newcommand{\removedtext}[1]{}

\newtheorem*{problem}{Problem}
\newtheorem{theorem}{Theorem}
\newtheorem{lemma}{Lemma}

\bibliographystyle{plainurl}

\begin{document}

\fi
\section{Additional details of the randomized \FPT-algorithm}
\subsection{Assumption on the word length}
\label{sec:wordlength}
\newtext {Our algorithm needs to store a table with $\Theta(2^kn)$
  entries, namely the entries $T_b(R,l,v)$ with fixed $b$ and~$l$.  We
  assume that these fit in memory and can be addressed in constant
  time.  Thus, it is reasonable to assume a word length of at least
  $k+\log n$ bits.
}

\subsection{Finite field calculations}
\label{sec:field}

Each evaluation of the recursion \eqref{eq:recursion} involves
 additions and multiplications in the
finite field~$\mathbb F_q$, where $q=2^s>n$ is a power of~2.

We will argue that
it is justified to regard the time for these arithmetic operations
as constant, both in theory and in practice.
The error probability can be controlled by choosing the
finite field sufficiently large, or by repeating the algorithm with
new random weights.

For implementing the algorithm in practice, arithmetic in
$\mathbb F_{2^s}$ is supported by a variety of powerful libraries.
see for example \cite{GF-complete}. %
The size that these libraries conveniently offer (e.g., $q=2^{32}$)
will be mostly sufficient for a satisfactory success probability.

For the theoretical analysis,
we propose to choose a finite field of size $q\ge n^2$, leading to a failure
property less than $1/n$, and to achieve further reductions of the
failure
property by repeating the algorithm $W$ times.

We will describe some elementary approach for setting up the finite
field
$\mathbb F_q$ of size $q\ge n^2$, considering that we can tolerate
a runtime and space requirement that is 
linear or a small polynomial in~$n$.
Various methods for arithmetic in finite fields $\mathbb F_{2^s}$ are
surveyed in Luo, Bowers, Oprea, and Xu~\cite{Finite-Fields-2012}
or
Plank, Greenan, and
Miller~\cite{plank13:finite-field-arithmetic}.
The natural way to represent elements of
$\mathbb F_{2^s}$ is as a
polynomial
modulo some fixed
irreducible polynomial $p(x)$ of degree $s$ over
$\mathbb F_{2}$.
The coefficients of the polynomial form a
bit string of $s$ bits.
Addition is simply an XOR of these bit strings.

We propose to use the
split table method,
 which is simple and implements multiplication by a few look-ups in small
 precomputed tables,
see
 \cite[Section~3.3.1]{Finite-Fields-2012}
or \cite[Section~6.5]{plank13:finite-field-arithmetic}.

Let
 $C=\lceil (\log_2 n)/2
 \rceil$, and let $q=2^{4C}$. Thus,
 $q \ge n^2$, as required.

 According to \cite[Theorem 20.2]{shoup-NT}, an irreducible polynomial $p(x)$
 of degree $d=4C$ over $\mathbb F_2$ can be found in
 $O(d^4)=O(\log^4 n)$ expected time, by testing random polynomials of
 degree $d$ for irreducibility.
(Shoup~\cite{shoup-NT} points out that this bound is not tight;
there are also faster methods.)

Multiplication of two polynomials modulo $p(x)$ can be carried out in
the straightforward ``high-school'' way in $O(d^2)=O(\log^2 n)$ steps,
by elementwise multiplication of the two polynomials, and reducing the
product by successive elimination of terms of degree larger than $d$.

We now consider the
bitstring of length $d$ as composed of $r=4$ chunks of size~$C$.
In other words, we write the polynomial $q(x)$ as
\begin{displaymath}
  q(x)=q_3(x)x^{3C}+q_2(x)x^{2C}+q_1(x)x^{C}+q_0(x),
\end{displaymath}
where $q_3,q_2,q_1,q_0$ are polynomials of degree less than~$C$.
Addition takes $O(r)$ time, assuming an XOR on words of length
$C=O(\log n)$ can be carried out in constant time. 
 Multiplication is carried out chunk-wise, using
 $2r-1=7$ multiplication tables.
The $j$-th table contains the products $r(x)r'(x)x^{jC}$, for all
 pairs $r(x),r'(x)$ of polynomials of degree less than~$C$, for $j=0,1,\ldots,2r-2$.
 Multiplication in the straightforward way takes then
 $O(r^3)$ time.
Each multiplication table has
 $2^C\times 2^C=O(n)$ entries
 of $r=4$ words, and it can be precomputed in $O(n d^2) = O(n\log^2 n)$
 time by multiplying in the straightforward way.

 Thus, after some initial overhead of
$ O(n\log^2 n)$ time, which is negligible in the context of the
overall algorithm, with $O(n)$ space, %
arithmetic in $\mathbb F_q$ can be carried out in $O(1)$ time in the
Word-RAM model, where arithmetic, logic, and addressing operations on words
with $\Theta(\log n)$ bits are considered as constant-time operations.

We mention that, with another representation, the initial setup can
even be made deterministic,
and its time reduced to
$ O(n)$ (keeping the
$ O(n)$ space bound), see Appendix~\ref{sec:deterministic-galois}.

\subsection{Reduction to polynomial space}
\label{sec:poly-space}

As described, the algorithm has an exponential factor $2^k$ in the
space requirement.
This exponential space requirement can be eliminated at the expense of
a moderate increase in the runtime,
by using
an inclusion-exclusion trick that was first used by Karp~\cite{karp-82} for the
Traveling Salesman Problem.
\newtext
{The possibility of applying this trick in the context of our problem was already mentioned in
Bj{\"o}rklund et al.~\cite{Bjork2012}.}

For a ``{forbidden} set'' $F\subseteq \{2,3,\ldots,k\}$, $v\in V$,
$1\le l\le n$, and $1\le j\le k$ 
let
$ U_b(F,l,j,v)$ be
the sum of the weights
of all walks that
\begin{compactitem}
\item start at $b$,
\item have their first edge in $S_1$,
\item end at $v$,
\item consist of $l$ edges,
\item contain no edge from the sets $S_i$ with $i\in F$,
\item use in total $j$ edges from the sets $S_i$, \newtext{counted with
  multiplicity:
  If an edge belonging to $p$ different sets $S_i$
  is traversed $r$ times, it contributes $pr$ towards the count~$j$.}
\end{compactitem}
The clue is that we can compute these quantities for each $F$
separately in polynomial time and space, by simply removing the
edges of $S_i$ for all $i\in F$ from the graph.
The dynamic-programming recursion is straightforward.
In contrast to \eqref{eq:recursion}, we have to keep track of the
number  $j$ of edges from $S_1\cup\cdots\cup S_k$.

We regard each index
$i\in \{2,\ldots,k\}$ as a ``feature'' that a walk might have
(or an ``event''), namely that it avoids the edges of $S_i$.
By the inclusion-exclusion formula, we can compute the
sum of weights of paths that have none of the features,
i.e., that visit \emph{all} sets $S_i$. In this way, we get the
following formula:
\begin{lemma}
\begin{equation}
  \label{eq:inc-exc}
  \hat T_b(l)
  =\sum_{F\subseteq \{2,3,\ldots,k\}} (-1)^{|F|}  U_b(F,l,k,b)
\end{equation} 
\end{lemma}
\begin{proof}
  The parameters $l$ and $b$ match on both sides.
  By the inclusion-exclusion theorem,
  the right-hand side is the sum of all walks in which every set $S_i$
  is visited \emph{at least} once.
  Since the parameter $j$ is equal to $k$, we know that %
  there were only
  $k$ visits to sets $S_i$; thus,
   every set $S_i$
  is visited \emph{exactly} once.
  \qed
\end{proof}

\newtext{The sign $(-1)^{|F|}$ is of course irrelevant over a field of
  characteristic~2.}

\newtext{
  To reduce the space requirement as much as possible,
  we organize the computation as follows.
  First, each starting point $b$ is considered separately.
  We initialize $n$ variables for accumulating the contributions
  to the quantities $\hat T_b(l)$ according to
  \eqref{eq:inc-exc}.
  Then, for each forbidden set $F$ separately,
  we compute the quantities
  $U_b(F,l,j,v)$ for all $j$ and $v$,
  incrementally increasing $l=1,2,\ldots ,n$.
  Since we need to remember only the entries for
  two consecutive values $l$ at a time,
  this requires only $O(nk)$ space.
  Along the way, we add the contributions
   $U_b(F,l,k,b)$ to \eqref{eq:inc-exc}.}

In summary, we can calculate
$\hat T_b(1),\hat T_b(2),\ldots$
in space $O(nk)$.
\newtext{Compared to the exponential-space algorithm,
}we need an additional factor~$k$ in the runtime,
  for the $k$
  choices of the parameter $j$.

\subsection{Recovering the solution}
\label{recover}

The algorithm, as described so far, works as an oracle that only gives a yes-no answer (and a length~$l$), but it
does not produce the solution.
To recover the solution, we will call the oracle
repeatedly with different inputs.

Suppose the algorithm was successful in the sense that some number
$ \hat T_b(l)$
turned out to be nonzero,
after finding only zero values for all smaller values of~$l$.
By 
Lemma~\ref{shortest-L-nonzero}, we conclude that there
exists a simple
cycle through $b$ satisfying the edge set constraints, possibly (with
small probability) shorter than $l$.

We will find this cycle by selectively deleting parts of the edges and
recomputing
$ \hat T_b(1),\hat T_b(2),\ldots,\allowbreak \hat T_b(l)$ for the reduced graph to
see whether this graph still contains a solution.
In this way, we will determine the successive edges of the cycle, and,
as we shall see,
 we will know the
cycle
after at most $4n\log_2 \frac {4m}n$ iterations.

The first edge out of $b$ is an edge of $S_1$, and thus we start by
looking for the first edge among these edges. (The other edge incident
to~$b$, by which we
eventually return to $b$, is not in $S_1$, and thus there is no
confusion
between the two edges incident to~$b$.)
In the general step, we have determined an initial part of the cycle
up to some vertex $u$, and we locate the outgoing edge among the edges
$(u,v)$ incident to $u$, excluding the edge leading to $u$ that we
have already used.

In general, for a vertex $u$ of degree $d_u$, we have a set of $d_u-1$
potential edges. We locate the correct edge by binary search: We split
the potential edges into two equal parts, and query whether a cycle
still exists when one or the other part is removed.
It may turn out (with small probability) that none of the two
subproblems
yields a positive answer.
In this case, we repeat the oracle with new
random weights. Since the success probability is greater than 1/2, we
are guaranteed to have a positive answer after at most two trials, in
expectation. 
\newtext
{(We mention that such repeated trials may be necessary only when the length $l$ for which
  we are looking is not the shortest length of a feasible cycle.
  Otherwise one can show,
  using arguments from the proof in Section~\ref{shortest-L-nonzero-proof},
  that
 the polynomial for the original problem is the sum of the polynomials
 for the two subproblems. Thus,
   at least one of the two
 subproblems
 must give a positive answer.)}
After at most $2\lceil \log_2 (d_u-1) \rceil$ successful queries, we have
narrowed down the search to a single outgoing edge $uv$, and we continue at the
next vertex~$v$.

For a walk $W$ of length $l$, we use, in expectation, less than
$Q := \sum_{u\in W}4(1+\log_2 d_u)$
queries, where
$\sum_{u\in W}d_u\le 2m$.
$Q$ can be bounded by
\begin{displaymath}
  Q \le 4l\left(1+\log_2 \tfrac{2m}l\right)
  \le 4n\left(1+\log_2 \tfrac{2m}n\right)
=   4n\log_2 \tfrac{4m}n,
\end{displaymath}
as claimed.

During this procedure, %
it may also turn out that a solution with fewer than $l$ edges exists.
In this case, we know that we must have been in the unlikely case that
the original algorithm failed to produce a nonzero value for the shortest solution $l$.
We simply adjust $l$ to the smaller value and continue.

\newtext{Note that this procedure is guaranteed to produce a simple cycle, although
not necessarily the shortest one.
 The algorithm selects a branch only when the corresponding polynomial
 is nonzero, which implies that a simple cycle exists in that branch.}

Some simplifications are possible.
For  edges that are known to belong to every solution (for example   
if some set $S_i$ contains only one edge), we can assign unit weight,
thus saving random bits, reducing the degree of the polynomial, and
increasing the success probability.
\newtext{Edges that would close a loop can be discarded.}
When the graph is sparse and $m=O(n)$, we can simply try the $d_u-1$ edges one
at a time instead of performing binary search, at no cost in terms of
the asymptotic runtime.

\subsection{Proof of Lemma~\ref{shortest-L-nonzero}}
\label{shortest-L-nonzero-proof}
(a) By assumption, there is a simple cycle
among the walks whose weights are collected in $\hat T_b(l)$, and
we easily see that the monomial corresponding to such a walk occurs
with coefficient~1. Hence  $\hat T_b(l)$ is not identically zero.
By definition,
  $\hat T_b(l)$ is a sum of weights of walks of length $l$, and hence it is clear that
it is homogeneous of degree $l$.

(b) For the second statement of the lemma,
we have to show that 
 $\hat T_b(l)$ is zero if there is no simple walk of length $l$ or
 shorter.
Since the field has characteristic~2, it suffices to establish a
matching among those closed walks that satisfy the edge set
constraints but don't represent simple cycles.
Let $W$ be such a walk. We will map $W$ to another walk $\phi(W)$ that uses the
same multiset of edges, by reversing (flipping) the order of the edges
of a subpath between two visits to the same vertex~$v$.
Our procedure closely follows Bj{\"o}rklund et al.~\cite{Bjork2012}, but we
correct an error in their description.

\begin{figure}[htb]
  \centering
\begin{align*}
  W = W_0 &= 123415651432345461786571 = 1[23415651432]345461786571\\
W_0'=  W_1 &= 1\underline 2345461786571 = 123[454]61786571
  \\
W_1'=  W_2  &  = 123\underline461786571 = 1234[61786]571
  \\[1ex]
W &= 1234156514323454[61786]571\\
\phi(W) &= 1234156514323454[68716]571
  \end{align*}  
  \caption{Mapping a nonsimple cycle $W$ to another cycle
    $\phi(W)$. The start vertex is $b=1$.}
  \label{fig:phi-W}
\end{figure}

An example of the procedure is shown in Figure~\ref{fig:phi-W}.
We look for the first vertex $v$ that occurs several times on the
walk. (During this whole procedure, the occurrence of $b$ at the start of
the walk is never considered.)
We look at the piece $[v\ldots v]$ between the first and last
occurrence of $v$ and flip it.
If the sequence of vertices in this piece is not a palindrome, we are done.
Otherwise, we cut out this piece from the walk.
The resulting walk will still visit an edge from each $S_i$ because
such an edge cannot be part of a palindrome, since it is visited only
once.
Since the resulting walk $W'$ is shorter than $l$, %
and thus shorter than $L$, by assumption, it cannot be a simple
cycle, and it must contain repeated vertices.

We proceed with $W'$
instead of $W$.
Eventually we must find a
piece $[v\ldots v]$ that is not a palindrome. We flip it in the
original walk $W$, and the result is the walk $\phi(W)$ to which $W$
is matched.
(The flipped piece
$[v\ldots v]$ does not necessarily start at the
first occurrence of $v$ in the original sequence $W$, because such an occurrence
might have been eliminated as part of a palindrome.\footnote
{Here our procedure differs from the method described in
  \cite{Bjork2012}, where the flipped subsequence extends between the first and
  last occurrence of~$v$ in~$W$.
  In this form, the mapping $\phi$ is not an
  involution. The proof in~\cite{Bjork2012}, however, applies to the method as described here.
  This inconsistency was confirmed %
  by the authors
  (%
  Nina Taslaman, private communication, February 2021).
})

It is important to note that
after cutting out a palindrome $[v\ldots v]$,
all repeated vertices must come \emph{after} the vertex $v$.
Hence, when the procedure is applied to $\phi(W)$, it will perform
exactly the same sequence of operations until the last step, where it
will flip $\phi(W)$ back to~$W$.

Since the first edge of the walk is unchanged, the condition that
this edge must belong to~$S_1$ is left intact.

This concludes the proof of
Lemma~\ref{shortest-L-nonzero}.
\qed %

\subsection{Other approaches}

The original algorithm of Bj{\"o}rklund et al.~\cite{Bjork2012} considers
simple paths through $k$ specified \emph{vertices} or (single) edges.
The treatment of vertex visits leads to some complications:
To make the argument for Lemma~\ref{shortest-L-nonzero} valid,
Bj{\"o}rklund et al.'s %
definition of allowed walks had to explicitly forbid ``palindromic
visits'' to the specified vertices, visiting the same edge twice in
succession (condition P4).  Consequently, the algorithm for
accumulating weight of walks has to take care of this technicality.

Another Monte Carlo FPT algorithm for the problem of finding a simple
cycle through $k$ specified vertices was given by
Wahlström~\cite{Wahlstrom13}. It runs in time $O(2^k\mathrm{poly}(n))$
in a graph with $n$ vertices.
This algorithm uses interesting algebraic techniques, and it is
also based on inclusion-exclusion, but it does not seem to extend %
to the case where one out of a \emph{set} of vertices (or edges) has
to be visited.

\section{Variation: Deterministic setup of finite field computations
  in characteristic 2}
\label{sec:deterministic-galois}

\subsection{Getting a primitive polynomial in quadratic time}
\label{sec:primitive-polynomial}

There is a completely naive algorithm for
constructing $\mathbb F_q$ in $O(q^2)$ time and $O(q)$ space,
for $q=2^s$, assuming $q$ fits in a word ($s$ bits). Simply
try out all polynomials $p(x)$
over $\mathbb F_2$
of degree less than $s$. There are $q$ possibilities.

For each $p(x)$, try to construct the logarithm table (``index
table'') in $O(q)$ steps by trying whether the polynomial $x$ generates the
nonzero elements of %
$\mathbb F_2[x]$:
Start with the string $00\ldots01$ representing the polynomial~$1$,
and multiply by $x$ by shifting to the left, and adding $p(x)$ (XOR
with the corresponding bit string) to clear the highest bit if
necessary.
Repeat $q-2$ times and check whether
$00\ldots01$ reappears. If not, then we are done.

(This will be successful iff $p(x)$ is a primitive polynomial modulo~$2$.
Actually, it is sufficient to check the powers $x^k$ of $x$ where $k$
is a maximal proper divisor of $q-1$. The primes dividing $q-1$ can
be trivially found in $O(\sqrt q)$ time, and there are less than
$\log q$ of them.
Computing~$x^k$ takes $O(s^2\log q)=O(\log^3q)$ time if the $O(s^2)$
schoolbook method of multiplication modulo $p(x)$ is applied, together
with repeated squaring to get the power.
Thus, $O(q\log^3q)$ time instead of $O(q^2)$. Probably a gross
overestimate because primitive polynomials modulo~2 are frequent; there are
$\phi(q-1)/s$ of them.)

\subsection{From \texorpdfstring{$\mathbb F_q$}{Fq} to
  \texorpdfstring{$\mathbb F_{q^2}$}{F(q**2)} in \texorpdfstring{$O(q)$}{Q(q)}
  deterministic time and space}
\label{sec:q2}

This is achieved by a degree-2 field extension.
We look for an irreducible polynomial of the form $p(x)=x^2+x+p_0$ over
$\mathbb F_q$.
If this polynomial were reducible, we could write
\begin{align*}
  p(x)=(x+a)(x+b)=
  x^2+(a+b)x+ab
\end{align*}
with $a+b=1$, hence $b=1+a$, and
\begin{align*}
  p(x)=(x+a)(x+(a+1))=
  x^2+x+a(a+1)
\end{align*}
The expression $a(a+1)$ can take at most $q/2$ different values,
because $a=c$ and $a=c+1$ lead to the same product, since $(c+1)+1=c$.
Thus, we can group the $q$ potential values $a$ into $q/2$ pairs with the
same product, and at least $q/2$ must be unused.
We can find the range of
$a(a+1)$ by marking its values in an array of size $q$. Any unmarked
value can be used as the constant term $p_0$. This takes $O(q)$ time
and space.

Multiplication of two polynomials
$a_1x+a_0$ and $b_1x+b_0$ gives the product
\begin{align*}
  c_1x+c_0&=(a_1x+a_0)(b_1x+b_0)
\\&
  =a_1b_1x^2 + (a_1b_0+a_0b_1)x+a_0b_0
\\&
  =a_1b_1x^2 + (a_1b_0+a_0b_1)x+a_0b_0 - a_1b_1(x^2+x+p_0)
\\&
  =(a_1b_0+a_0b_1+a_1b_1)x+(a_0b_0+a_1b_1p_0)
\end{align*}
Multiplication can therefore be carried out as follows:
\begin{align*}
  c_1 &= a_1b_0+a_0b_1+a_1b_1
  \\&
  =(a_1+a_0)(b_1+b_0)-a_0b_0
  \\
c_0&=  a_0b_0+a_1b_1p_0
\end{align*}
with four multiplications in $\mathbb F_q$ (and four additions), using the common term
$a_0b_0$ for both coefficients.\footnote
{This is the same trick as for multiplying two complex numbers with
three multiplications instead of four. By contrast, \cite[Section
6.8]{plank13:finite-field-arithmetic} %
propose irreducible polynomials of the form $x^2+p_1x+1$, leading to
five multiplications. On the other hand, in cases where multiplications
are done by index tables, the mere number of multiplication operations
does not determine the runtime alone; a common factor that appears in several
multiplications saves lookup time in the logarithm tables.}

\subsection{Setup of a finite field \texorpdfstring{$\mathbb F_q$}{Fq}
  with
  \texorpdfstring{$q>n^2$}{q>n**2}
  in
  deterministic \texorpdfstring{$O(n)$}{O(n)}
  time and space}

Let
 $C=\lceil (\log_2 n)/2
 \rceil$.
First we construct $\mathbb F_{q}$ for  $q=2^{C}=O(\sqrt n)$ in $O(q^2)=O(n)$
time, as described in Section~\ref{sec:primitive-polynomial}.
To carry out multiplications in $\mathbb F_{q}$ in constant time, we can store
a logarithm and antilogarithm table, in $O(q)$ space, or we can even
compute a complete multiplication table, in $O(q^2)=O(n)$ time and space.

Then, by the method of Section~\ref{sec:q2},
we go from
$q=2^{C}$ to $q=2^{2C}$, and finally from
$q=2^{2C}$ to
$q=2^{4C}$, in $O(2^{2C})=O(n)$ time and space.

Then addition and multiplication in 
$\mathbb F_{2^{4C}}$ can be carried out in constant time.
Multiplication goes down two recursive levels, from 
$q=2^{4C}$ via $q=2^{2C}$ to
$q=2^{C}$, before the 16 resulting multiplications are resolved by table look-up.

If we prefer, we can eliminate the lower level of recursion by building a log/antilog table for
$q=2^{2C}=O(n)$ of size $O(n)$, to do the multiplication by table
look-up already at this level. To prepare
the tables, we need a generating element of 
$\mathbb F_{2^{2C}}$. Such a generating element can be constructed in
$O(2^{2C})=O(n)$ time,
 as shown in the next section~\ref{cycle-generator}.

\subsection{Constructing an index table without knowing a
  generator}
\label{cycle-generator}

We assume that multiplication in $\mathbb F_q$ takes constant time,
but no generating element for the multiplicative group of
$\mathbb F_q$ is given.  Our goal is to construct logarithm and
antilogarithm tables (index tables), of size $q-1$, in time $O(q)$.  In
order to generate these tables, we need a generating element.

Essentially, we have to find a generator of a cyclic group, which is
given by a multiplication oracle and the list of its elements.

We start with an arbitrary %
element $a_1\ne 0$ and try to construct the
index table by running through the powers of $a_1$. We mark
the elements that we find.
If the process returns to $1$ before marking all nonzero elements,
we pick an unmarked element $a_2$ and run the same process with $a_2$.

If $a_1$ has order $o_1$ in the multiplicative group (which is cyclic of
order $q-1$),
and $a_2$ has order~$o_2$, we show how to construct an element $a$ of
order $o=\mathrm{lcm}(o_1,o_2)$, which generates the subgroup
$\langle a_1,a_2\rangle$,
as follows:

First we ensure that $o_1$ and $o_2$ are relatively prime, without
changing $\mathrm{lcm}(o_1,o_2)$.
For each common prime divisor $p$ of $o_1$ and $o_2$, we determine the
largest power of $p$ dividing $o_1$ and $o_2$, respectively:
$p^{f_1}|o_1$
and $p^{f_2}|o_2$.
If $f_1\le f_2$ we
replace $a_1$ by $(a_1)^{p^{f_1}}$, dividing the order $o_1$  of $a_1$
by $p^{f_1}$; otherwise we proceed analogously with $a_2$.

After this preparation, $o=\mathrm{lcm}(o_1,o_2)=o_1o_2$.
Let $g$ be a generator of the group 
$\langle a_1,a_2\rangle = \{1,g,g^2,\ldots,g^{o-1}\}$.
Then 
$\langle a_1\rangle$ is generated by $g^{o/o_1} = g ^{o_2}$,
and we know that
$a_1 = g ^{j_1o_2}$ for some $j_1$. %
Similarly,
$\langle a_2\rangle$ is generated by $g^{o/o_2} = g ^{o_1}$,
and %
$a_2 = g ^{j_2o_1}$ for some~$j_2$. %

By the extended Eulerian algorithm for calculating the greatest
common divisor of $o_1$ and $o_2$, we find $u_1,u_2$ such that
\begin{displaymath}
  o_1u_1+o_2u_2 = 1.
\end{displaymath}
We claim that $a := a_1^{u_2}  a_2^{u_1}$ generates
$\langle a_1,a_2\rangle$.
To see this, we show that $a_1$ and $a_2$ are powers of $a$,
in particular, $a^{o_2}=a_1$ and  $a^{o_1}=a_2$:
\begin{align*}
  a^{o_2}  
    =
    (a_1^{u_2}  a_2^{u_1})^{o_2}
  &=
  g^{(j_1o_2u_2+j_2o_1u_1)o_2}
  \\&
  =  g^{(j_1(1-o_1u_1)+j_2o_1u_2)o_2}
  \\&
  =  g^{j_1o_2-j_1o_1u_1o_2+j_2o_1u_2o_2}
  \\&
  =  g^{j_1o_2} (g^{o_1o_2})^{-j_1u_1+j_2u_2}
  =  a_1\cdot 1^{-j_1u_1+j_2u_2} 
  = a_1
\end{align*}
The proof of the equation $a^{o_1}=a_2$ is analogous.

In summary, the above procedure shows how we get from
an element $a_1$ generating a subgroup
$\langle a_1\rangle$
of
order $o_1$ 
and an element $a_2$ that is not in that subgroup to
an element $a$ generating a strictly larger subgroup
$\langle a\rangle = \langle a_1,a_2\rangle$,
whose order $o$ is a multiple of $o_1$. The procedure can be carried
out in $O(o)$ time, assuming an array of size $q-1$ (corresponding to
the whole group) is available.
We repeat this procedure until we have found a generating element for
the whole group.
Since the size~$o$ is at least doubled in each step,
the procedure can be carried out
in $O(q)$ time and space in total.

\ifstandalone
   \let\next\relax
\else
   \let\next 
\fi
\next

\bibliography{../EuroCG22/impop}

\section{Collected Stuff about finited fields and primitive
  polynomials}
\label{sec:primitive}

The most convenient for preparing a table of logarithms are
\emph{primitive polynomials} $p(x)\in \mathbb F_2[x]$:
irreducible polynomials of degree $m$ for which the powers $1,x,x^2,x^3,\ldots$
generate all $2^m-1$ nonzero elements of
$ \mathbb F_2[x]/p(x)$.
In other words, $x$ is a primitive element.

1. The number of primitive polynomials is $\phi(2^m-1)/m$.

\begin{displaymath}
  \phi(n) \ge \frac n
  { e^\gamma \ln \ln n +
\frac{2.50637}
    {\ln \ln n}}
\end{displaymath}
for $n\ge 3$,
where $\gamma \approx0.5772$
is Euler's constant.

Approximate formulas for some functions of prime numbers, J.B. Rosser and L. Schoenfeld, Illinois J. Math. 6 (1962) 64--95.
Formula (3.42) on p. 72

Note: actually
  $+\frac5 {2 \ln \ln n}$ with the exception of
$n=223092870=2.3.5.7.11.13.17.19.23$, formula (3.41)

Consequence: The probability that a random polynomial of degree $d$
over $\mathbb F_2$ is primitive is at least
$\Omega(1/(d \log d))$.
(If $2^d-1$ is a prime, then every irreducible polynomial is
primitive.
Probability of being irreducible is at least $1/2d$.)

Tables of Finite Fields
Author(s): J. D. Alanen and Donald E. Knuth
Source: Sankhyā: The Indian Journal of Statistics, Series A
(1961-2002), Dec., 1964, Vol. 26, No. 4 (Dec., 1964), pp. 305-328
Published by: Indian Statistical Institute
Stable URL: \url{https://www.jstor.org/stable/25049338}

They speak of
``indexing polynomial'' instead of ``primitive polynomial''
because  ``primitive polynomial'' is sometimes (for example in Victor
Shoup's book) used for a polynomial
where the greatest common divisor of the coefficients is 1. (Gauss's
Lemma says that the product of two primitive polynomials (in the
latter sense) is primitive.)

Knuth TAOCP Vol.2 \S4.6.1 p.422, distinguishes
``primitive polynomial'' (whose coefficients are relatively prime)
from ``primitive polynomial modulo $p$'' discussed in \S3.2.2, p.30.

2. With a primitive (or ``indexing'') polynomial, generation of the
logarithm table becomes very simple.
multiplication by $x$ is just a left shift by 1 position, possibly
followed by a XOR with $p(x)$ to clear the highest bit.

$O(1)$ on a Word-RAM.

A Computational Introduction to Number Theory and Algebra (Version 2) Victor Shoup
\url{https://shoup.net/ntb}

Theorem 20.2.Algorithm RIP uses an expected number
ofO(4len(q))opera-tions inF, and its output is uniformly distributed
over all monic irreducibles ofdegree.

for $q=2$.

$O(d^4)$ for finding a random irreducible polynomial of degree $d$
over $\mathbb F_2$.

Shoup points out that this runtine is an over-estimate.

``the application to the analysis of Algorithm RIP, is essentially due
to Ben-Or [14, 1981]. If one implements Algorithm RIP using fast polynomial
arithmetic,  one gets an expected cost of [slightly superquadratic in $\ell$]
operations in F.   Note that Ben-Or’s analysisis a bit incomplete ---
see Exercise 32 in Chapter 7 of Bach and Shallit [11] for a complete
analysis of Ben-Or’s claims.''

[11, p.191--192]:
irreducibility test.
The expected number of bit operations
for finding an irreducible polynomial is $O(d\lg d \,\lg q \,(\lg f)^2)$.
$f$ is the polynomial. The field is $\mathbb F_q$.
$\lg f = (d+1)\lg q$ is defined on p.132, it is the number of bits in
a standard dense representation (and does not depend on the value of $f$).
The factor $(\lg f)^2$ is the time for multiplication (and division, if defined) mod $f$.

Solution on p.359.
 
\hrule
STUFF FROM PREVIOUS VERSIONS.

Multiplication can be reduced to addition (modulo $q-1$)
via a table of logarithms and antilogarithms of size $q$
cite[Chapter 10, Table~A]{Lidl-Niederreiter}.
Also [Alanen and Knuth].

If $q$ is chosen between $2n$ and $4n$, just barely ensuring
a success probability of at least $1/2$, then such tables
can be stored in $O(n)$ words of $s=O(\log n)$ bits,
and they can be prepared in $O(n\times s) =
O(n\log n)
$ time if a primitive polynomial is used.
 This is by far dominated by the other parts of the algorithm.
For all reasonable values of $s$,
 primitive polynomials of degree~$s$ over
 $\mathbb F_{2}$ are tabulated and readily
 available.

 Split table method
 \cite[Section~3.3.1]{Finite-Fields-2012}
 
 An alternative, which allows larger orders $q=2^s$, is to
 choose a small constant $r\ge 2$
 and consider
 bitstrings composed of $r$ chunks of size
 $C=\lceil (\log_2 n)/2
 \rceil$.
Addition takes $O(r)$ time, assuming an XOR on words of length
$C=O(\log n)$ can be carried out in constant time. 
 Multiplication is carried out chunk-wise, using
 $2r-1$ multiplication tables, each with 
 $2^C\times 2^C=O(n)$ entries
 of $r$ words.
 Multiplication can then be carried out in a straightforward way in
 $O(r^3)$ time.

In this way, one can achieve a field or order $q\ge n^{r/2}$
and reduce the failure probability to $q/n\le n^{1-r/2}$.
 The tables need $O(r^3n)$ words of storage, and they can be prepared
 in [CHECK THIS!]
 $O(r^3n\log(rn))$ time. For small constant $r$, all of this is
 negligible compared to the requirements of the remaining
 algorithm.

 (use PRIMITIVE polynomial!) 

 Alternative

 $\mathbb F_{2^{Cr}}$
as a finite field extension of degree $r$ over the base field
$\mathbb F_{2^{C}}$
and use table-lookup techniques for the base field.
Multiplication in $O(r^2)$ time.

see the survey \cite{Finite-Fields-2012}

 However, since the effort for finite field operations increases
 more than linearly with the logarithm of the order of the field, 
 the most efficient way to reduce the failure probability beyond $1/n$
 is to use a fixed number, say $r=4$, or chunks and carry out $W$
 independent random trials. The time and space is multiplied by a
 factor $W$, and the failure probability goes down to $1/n^W$.
 For this choice of parameters, we formulate the resulting runtime as follows:

 begin{theorem}
   For any $W\ge 1$
   $O(W$
Monte Carlo, NOT Las Vegas
   $1-1/n^W$
 end{theorem}
 
 \url {https://www.ssrc.ucsc.edu/Papers/plank-fast13.pdf}
 Screaming Fast Galois Field Arithmetic Using Intel SIMD Extensions

GR: [ Maybe this is too detailed for the laymen but at the same time
too ignorant for the experts.]

Joachim von zur Gathen and Jürgen Gerhard,
Modern Computer Algebra, third edition, Cambridge University Press, 2013.

\end{document}

\section{Adapting the Edge Set Constraints}
\label{sec:special-edge-sets}

\newtext{
In Section~\ref {sec:cycle-edge} we have outlined
a modified construction of the graph $G$ on which a nonintersecting
Eulerian cycle is sought,
which avoids the quadratic blow-up of the number of edges.
}

\newtext{
Here we give more details of this construction, see
Figure~\ref{fig:one-bad-face_c}.
We cut off each run of
consecutive %
edges, like $\textit{fghi}$, by an additional edge (shown
dotted in
Figure~\ref{fig:one-bad-face_c}) and place a single terminal node
there.}

\newtext{
One has to take care that the cycle that is found does not use two such
terminal edges in succession, like the edges crossing $f$ and $h$,
because such a cycle would not correspond to a valid curve. This
constraint must be added to the problem definition, and the
recursion~\eqref{eq:recursion} %
must be modified accordingly.}

The sets $S_i$ have %
a special structure.
Each set $S_i$ consists of a few
vertex-disjoint edges (the thick %
edges in
Figure~\ref{fig:special-edge-sets}), which we call \emph{central
edges}.
We call the edges connecting the central edges to the remainder of the graph the \emph{peripheral edges}.
We don't want to consider cycles that use two such peripheral edges (of the same
$S_i$) in succession, without going through the central edge.
We impose this as an extra condition on the walks whose weights we
accumulate.

\begin{figure}[htb]
\centering
\includegraphics{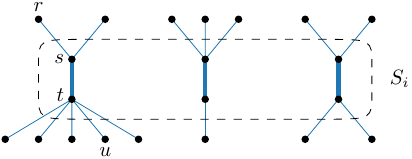}  
\caption{A set $S_i$ and its connecting edges}
\label{fig:special-edge-sets}
\end{figure}

It is easy to incorporate this condition in the
recursion of Section~\ref{dyn-prog}:
Instead of the quantities
$ T_b(R,l,v)$, we work with quantities
$ \tilde T_b(R,l,v,p)$ that depend \newtext{on} an additional parameter~$p$.
This is one bit
that tells whether the last edge of the walk has traversed a peripheral edge
in the direction towards the central edge. If this is the case, we
force the walk to use the central edge in the next step.

In this way, the additional condition incurs a blow-up of at most a
factor 2 in the size of the dynamic programming tables and in the runtime.

\indent\newtext{
We now argue that
Lemma~\ref{shortest-L-nonzero} still holds for these modified
quantities.
The proof in
Section~\ref{shortest-L-nonzero-proof} goes through for the following
reason. The conditions on allowed walks ensure that
an %
endpoint of a central edge,
like the vertex~$s$ in Figure~\ref{fig:special-edge-sets}), is
always part of a subpath consisting of a central edge
surrounded by two peripheral edges,
like $(r,s,t,u)$ or
$(u,t,s,r)$.
Such a subpath cannot be part of a palindrome,
since $(s,t)$ belongs to a special set $S_i$,
and for the same reason, $s$ can never be a repeated vertex of a walk.
If the bijection constructed in the proof reverses
a subpath that goes through
the vertex $s$, this is no problem because
the reversed traversal
does not violate %
the extra condition.%
}
\end{document}